\documentclass[a4paper,onecolumn,11pt,accepted=2024-04-24]{quantumarticle}

\pdfoutput=1 

\usepackage{graphicx}
\usepackage{amsmath}
\usepackage{amssymb}
\usepackage{amsfonts}
\usepackage{amsthm}
\usepackage{mathtools}
\usepackage{hyperref}

\usepackage{xcolor}
\usepackage[T1]{fontenc}
\usepackage[english]{babel}
\usepackage{braket}
\usepackage{enumitem}
\usepackage{tikz}
\usepackage{dsfont}
\usepackage{bbm}
\usepackage{physics}
\usepackage{url}

\usepackage{appendix}

\hypersetup{colorlinks=true,urlcolor=[rgb]{0,0,0.5},citecolor=[rgb]{0.5,0,0},linkcolor=[rgb]{0,0,0.4}}

\usepackage[bbgreekl]{mathbbol}
\DeclareSymbolFontAlphabet{\mathbb}{AMSb}
\DeclareSymbolFontAlphabet{\mathbbl}{bbold}
\usepackage{tikz-cd}
\usepackage{forest}
\usepackage[capitalise]{cleveref}

\theoremstyle{plain}

\newtheorem{thm}{Theorem}

\newtheorem{lem}{Lemma}
\newtheorem{cor}{Corollary}
\newtheorem{prop}{Proposition}

\newtheorem{defi}{Definition}
\newtheorem{conj}{Conjecture}

\theoremstyle{remark}
\newtheorem{rem}{Remark}

\definecolor{Green}{HTML}{00AD69}  
\definecolor{coolblue}{RGB}{0,51,102}
\definecolor{lightblue}{RGB}{102,210,255}
\definecolor{lightpurple}{RGB}{140,30,255}
\definecolor{lightpink}{RGB}{204,0,204}
\definecolor{midblue}{RGB}{0,102,204}
\definecolor{midpink}{RGB}{153,0,153}
\definecolor{darkblue}{RGB}{0,0,153}
\definecolor{cyan}{RGB}{0,204,204}
\definecolor{lightgreen}{RGB}{0,255,128}
\definecolor{midgreen}{RGB}{0,204,0}
\definecolor{midyellow}{RGB}{204,204,0}
\definecolor{darkyellow}{RGB}{153,153,0}
\definecolor{darkpurple}{RGB}{102,0,102}
\definecolor{orange}{RGB}{255,153,51}
\definecolor{darkred}{RGB}{153,0,76}
\definecolor{lightyellow}{RGB}{255,255,153}
\definecolor{lightred}{RGB}{255,153,153}

\newcommand{\bbC}{\mathbb{C}}

\newcommand{\ttLie}{\mathtt{Lie}}
\newcommand{\frakg}{\mathfrak{g}}

\newcommand{\bbR}{\mathbb{R}}

\newcommand{\bbN}{\mathbb{N}}

\newcommand{\calA}{\mathcal{A}}

\newcommand{\fraksu}{\mathfrak{su}}

\newcommand{\wiel}{\mathtt{Wie}\ell}
\newcommand{\1}{\mathds{1}}

\newcommand{\Pol}{\text{Pol}}
\newcommand{\st}{ : }
\newcommand{\T}{\mathbb{T}}
\newcommand{\C}{\mathbb{C}}

\usepackage{comment}
\newif\ifshow
\showfalse 
\ifshow
\includecomment{hide}
\else
\excludecomment{hide} 
\fi

\begin{document}

\title{A generic quantum Wielandt's inequality}

\author{Yifan Jia}
\affiliation{Department of Mathematics, Technische Universit\"at M\"unchen, Germany}
\affiliation{Munich Center for Quantum Science and Technology (MCQST), Germany}

\author{{\'A}ngela Capel}
\affiliation{Fachbereich Mathematik, Universit\"at T\"ubingen, Germany}


\begin{abstract}\noindent
Quantum Wielandt's inequality gives an optimal upper bound on the minimal length $k$ such that length-$k$ products of elements in a generating system span $M_n(\bbC)$. It is conjectured that  $k$ should be of order $\mathcal{O}(n^2)$ in general. In this paper, we give an overview of how the question has been studied in the literature so far and its relation to a classical question in linear algebra, namely the length of the algebra $M_n(\bbC)$. We provide a generic version of quantum Wielandt's inequality, which gives the optimal length with probability one.  More specifically, we prove based on \cite{KlepSpenko_SweepingWords_2016} that $k$ generically is of order $\Theta(\log n)$, as opposed to the general case, in which the best bound to date is $\mathcal O(n^2 \log n)$.  Our result implies a new bound on the primitivity index of a random quantum channel. Furthermore, we shed new light on a long-standing open problem for Projected Entangled Pair State, by concluding that almost any translation-invariant PEPS (in particular, Matrix Product State) with periodic boundary conditions on a grid with side length of order $\Omega( \log n )$ is the unique ground state of a local Hamiltonian. We observe similar characteristics for matrix Lie algebras and provide numerical results for random Lie-generating systems.
\end{abstract}

\maketitle

\tableofcontents


\section{Introduction}

Over the last twenty years, tensor network representation has been one of the most useful tools from quantum information theory for investigating quantum many-body systems. It has been discovered, due to an area law for entanglement, that the states described by 1D arrays of tensors - namely the Matrix Product States (MPS) - are well-suited for efficiently describing ground states of gapped 1D local Hamiltonians \cite{hastings2006solving}. Conversely, to guarantee that a certain MPS is the unique ground state of a certain gapped local Hamiltonian, the corresponding region of the network should be large enough to ensure ``injectivity'' \cite{perez2006matrix}. To translate this to a mathematical problem, let us fix the matrices that determine the tensor for a translation-invariant MPS with periodic boundary conditions. Then the chain of sites is called injective when the number of sites $L$ is large enough such that the length-$L$ products of the matrices span the whole matrix algebra, i.e. 
\begin{equation}\label{def:Wie-generating system}
   M_n(\mathbb{C}) = \text{span}\left\lbrace A_1 \ldots A_L \, |\, A_i\in S \text{ for all } i\in[L] \right\rbrace   
\end{equation}
for $S$ consisting of the $g$ site-independent $n$-by-$n$ matrices. Note that $g$ refers to the dimension of the space associated with each site and $n$ is the bond dimension. 

This is the first motivation for introducing a mathematical problem named \textit{quantum Wielandt's inequality}. If a subset $S$ satisfies Eq. (\ref{def:Wie-generating system}) for some well-chosen $L\in\bbN$, then the set is called a \textit{(Wie-)generating system} of $M_n(\bbC)$. In this case, we refer to the \textit{Wie-length} of a (Wie-)generating system $S$ as the minimal $L$ satisfying Eq. (\ref{def:Wie-generating system}), denoted by $\wiel(S)$.
\begin{hide}
    \begin{equation}
    \wiel(S):=\min \{L| M_n(\mathbb{C}) = \text{span}\left\lbrace A_1 \ldots A_L \, ,A_i\in S  \right\rbrace\}.
\end{equation}
\end{hide}
It is clear from its definition that Wie-length is finite when it is well-defined. Furthermore, it is reasonable to quantify the minimum, because it holds for any $m\geq \wiel(S)$ that all products of matrices of $S$ with length $m$ span the whole matrix algebra (cf. \cite{perez2010characterizing}). 

In fact, the name of Wielandt's inequality arises from a classical information-theoretic problem \cite{wielandt1950unzerlegbare}, namely how many repetitions are required for a certain memory-less channel so that any probability distribution after transmission has non-zero probability for all events. Translated to the quantum case, similarly, one can explore the minimal number of repetitions required for a certain quantum channel to transmit any density operator to a density operator with only positive eigenvalues \cite{sanz2010quantum}. Using the Kraus representation of the quantum channel, this problem in the quantum version can be reformulated as calculating the \textit{Wie-length} of the Kraus operators of the channel, so that the channel at that time point would be eventually of full Kraus rank.

The object that we denote here by \textit{Wie-length}, coined in that way from its relation to quantum Wielandt's inequality and the comparison to the \textit{length} of a matrix algebra, has also been considered before in the literature with different names that are more related to its physical applications mentioned before, e.g. \textit{injectivity index} \cite{Cirac_2019} or \textit{injectivity length} \cite{Molnar_2018} of a normal MPS, or combined together with the \textit{primitivity index} of a quantum channel. While the implication of the quantum Wielandt's inequality in the latter two topics is already well-studied in \cite{perez2006matrix} and \cite{sanz2010quantum}, the remaining open question is to find an optimal (uniform) upper bound on the Wie-length in general. This problem has been included as one of the main open problems in \cite{Cirac_2019}.  

More specifically, it has been proved that there exists a function $f(n)$ that only depends on the dimension $n$, such that the function gives an upper bound on the Wie-length of any Wie-generating system, no matter how we choose the subset from the matrix algebra $M_n(\bbC)$. \textit{Quantum Wielandt's inquality} conjectures that the optimal such function would be asymptotically of order $O(n^2)$ \cite{sanz2010quantum}. The first general result in this direction  was that $\wiel(S)$ is of order $O(n^4)$ for any generating system of $M_n(\bbC)$ \cite{sanz2010quantum}. This upper bound was subsequently improved in \cite{michalek2019quantum} to $O(n^2\log n)$, and this is the best general result to date, which still has a $\log n$-overlap compared to the conjectured optimal uniform upper bound.

A related definition in linear algebra is namely the \textit{length} of 
finite-dimensional algebra. In this problem, we quantify the length of a generating system $S\subset M_n(\bbC)$ as the minimal $\ell\in\bbN$ such that the products with length less than or equal to $\ell$ contain a basis for the matrix algebra $M_n(\bbC)$. The difference here with respect to $\wiel$ is that we do not require all the basis elements to be of the same length. There exists also an unproven long-standing conjecture here about the optimal uniform upper bound, which is expected to be of order $ O (n)$ \cite{Paz_Conjecture_1984}. The best proven result till the date has a $\log n$-overlap, namely it is of order $O(n\log n)$ \cite{shitov2019improved}.

Apart from the worst-case scenarios, both problems above can be studied in the \textit{generic case} (i.e. with probability $1$). It was conjectured in \cite{perez2006matrix}, and also can be checked numerically, that for a randomly chosen generating system, the Wie-length scales as $\Theta (\log n)$. We should expect that, by increasing the length by $1$, a new word that is linearly independent to the subspace spanned by the previously existing products is generated, as long as the dimension does not saturate.

In this paper, we first review the problems of estimating the length and the Wie-length of matrix algebras, namely the so-called \textit{quantum Wielandt's inequality} and \textit{Paz's conjecture}, and we discuss the relationship between them. Moreover, we provide a complete proof of the generic version of quantum Wielandt's inequality for any dimension $n$ with no restriction on the cardinality of the generating system by adapting an existing proof for the length of a generic generating system in \cite{KlepSpenko_SweepingWords_2016}. This result is written in a different language, so here we adapt it to a setting closer to quantum information theory. We notice that the order $O(\log n)$ and the constant factor in the new proof are both optimal, which improves the best generic bound proven so far in \cite{cadarso2013wielandt}. 


As a direct consequence, we indicate how this optimal generic bound applies to MPS and quantum channels and improves results of relevance in quantum information theory. In particular, we show that an MPS of length $\Omega(\log n)$ is almost surely the unique ground state of a local Hamiltonian, and that the index of primitivity of a generic quantum channel is of order $\Theta(\log n)$. A detailed proof of the latter result is included in an appendix written by Guillaume Aubrun and Jing Bai. 

Moreover, we generalize the result for MPS to higher dimensions and obtain interesting consequences in the context of injectivity of PEPS, namely that almost all translation invariant Projected Entangled Pair States (PEPS) with periodic boundary conditions on a grid with side length of $\Omega(\log n)$ are the unique ground state of a local Hamiltonian, irrespective of the dimension of the grid. This gives new insight into the generic case of some of the most relevant open problems for PEPS, namely Questions 4 and 5 in \cite{Cirac_2019}.

In the last part of our paper, we generalize the questions posed above to the context of finite-dimensional Lie algebras. We introduce a conjecture for the upper bound on the length of a generating system for some typical   matrix Lie algebra in the generic case and include numerical simulation to support our conjecture. Moreover, we indicate how this generic result for Lie algebras would apply to mathematical models of some physical problems, especially also to gate implementation of quantum computers.

\section{The Length of a Matrix Algebra}

In this section, we will present a short overview on a well-studied problem in linear algebra concerning the length of a generating system to generate a matrix algebra, and the conjecture on the optimal length, known as \textit{Paz's conjecture}.  

One natural way to construct a basis of a finite-dimensional algebra from a subset is to build up products (\textit{words}) from the elementary elements (\textit{letters}) in the subset (\textit{alphabet}). The first question that arises from this setting is whether the `alphabet' is good enough so that a basis is achievable, which results in the definition of \emph{generating system}. This can be determined by Burnside's theorem \cite{lomonosov2004simplest} for matrix algebras. Furthermore, the \emph{length of the generating system} is defined as the minimum required length for the longest word in the basis. The latter question connects naturally to the complexity of some algorithms as it counts the number of multiplications performed. 

Moreover, taking the maximum over all generating systems gives the \emph{length of the algebra}, which is thought to constitute a numerical characteristic for verifying finite algebras \cite{al2000reducibility,al2003unitary}. However, even for matrix algebra $M_n(\mathbb{C})$, a complete proof for a general optimal estimate of the length remains unknown.

Before presenting an overview of the results of the literature related to the aforementioned conjecture on the length of a matrix algebra, let us formally introduce all the notions mentioned above, in the framework of a complex matrix algebra $M_n(\bbC)$. 

 \begin{defi}\label{def:length}
 Let $S$ be a finite subset of $M_n(\bbC)$. We say that $S$ is a \emph{generating system} (GS) if there exists $\ell \in \mathbb{N}$ such that  
 \begin{equation*}
     M_n(\bbC)=span \left\{A_1A_2\dots  A_\ell \, | \, A_i\in\{S,\mathds{1}\} \text{ for all }i\in[\ell] \right\} \, .
 \end{equation*}
  The minimum $\ell \in \mathbb{N}$ for which this holds is called the \emph{length} of $S$ and we denote it by $\ell(S)$.   

  The \emph{length} of the algebra $M_n(\bbC)$ is the maximal length of all generating systems of $M_n(\bbC)$:
     \begin{equation*}
         \ell(M_n(\bbC))=\max_{S: \ell(S)< \infty}\ell(S).
     \end{equation*}
     The maximum exists since $\ell\in\mathbb{N}$ and $\ell (S)\leq \dim(\calA)$ for all generating systems.
 \end{defi}

\subsection{Overview on Paz's Conjecture}\label{subsec:Paz}

In 1984, Paz conjectured that the required length to generate the whole matrix algebra $M_n(\bbC)$, in the worst case scenario, should be $2n-2$  \cite{Paz_Conjecture_1984}. 

\begin{conj}[Paz's Conjecture]\label{conj:Paz}
  $\ell(M_n(\mathbb{C}))= 2n-2$. 
\end{conj}

The conjectured length cannot be improved in general, since there are examples of generating sets with exact length $2n-2$, for example for generating pairs including one rank-one matrix  \cite{LogstaffRosenthal_LengthsIrreduciblePairs_2011}. For small dimensions, the conjecture was initially proven to hold for $n\leq 5$ in \cite{LongstaffNiemeyerPanaia_LengthDimension5_2006}, and subsequently lifted to $n=6$ in \cite{LambrouLongstaff_LengthDimension6_2009}. 
However, in the low-dimensional case the subwords can still be checked concretely, which is not applicable for higher dimensions. In fact, the only known general upper bounds for the length of the matrix algebra have been for years of order $O(n^2)$  \cite{Paz_Conjecture_1984} and $O(n^{3/2})$ \cite{Pappacena_UpperBound_1997}, respectively. The key ingredient in the proofs are some combinatorial lemmas on words and the Cayley-Hamilton theorem.

The recent \cite{shitov2019improved} improves the technique in \cite{Pappacena_UpperBound_1997} and shows the tightest result to the date, which only includes a logarithmic overhead, namely $O(n\log n)$. However, a general complete proof for a linear upper bound on the length of $M_n(\bbC)$ still remains an open problem.

Nevertheless, there are some more examples, additional to the ones presented above, for which the conjecture is known to hold. Indeed, for arbitrary $n$, the tight upper bound $2n-2$ has been additionally shown to be achieved for generating systems with specific properties, such as having a matrix with distinct eigenvalues \cite{Pappacena_UpperBound_1997},  or more generally for pairs including a non-derogatory matrix \cite{GutermanLaffeyMarkovaSmigoc_PazConjecture_2018}, meaning that the minimal polynomial of the matrix has the degree that is equal to the dimension of the matrix. A great advantage that these properties present is that they can be used to create a basis, which is taken to be the canonical generating system, and thus one only needs to optimize the upper bound on the length of such generating systems. Note that this kind of property appears for almost all generating systems. However, it provides information on the opposite direction to the previous goal of finding the length of the \textit{worst-case scenario}. Indeed, as one can readily verify from numerical experiments, the generic case does not require a bound of $2n-2$ for the length, but rather only a logarithmic dependence of $n$. This is supported with an analytical proof derived in \cite{KlepSpenko_SweepingWords_2016},  where an upper bound of order $\Theta(\log n)$ on the length of a generic generating system has been verified with arguments based on constructions by graphs.

\section{Quantum Wielandt's Inequality}

For many applications in several contexts, it is more relevant to discuss another notion of length, which we coin \textit{Wie-length} after the problem that motivates it: Quantum Wielandt's inequality. The only difference with the latter length of an algebra is that, for the Wie-length, we require all words in the basis to be of same length.

In this section, we will present the definition for Wie-length and discuss a few of its properties, such as the fact that the Wie-length of a generating system stabilizes (i.e. if words of length $l$ generate the matrix algebra, then words of length $m\geq l$ also do), and the equivalence between the notions of generating system and Wie-generating system (which in turn allows us to conveniently exchange them throughout the text). As a consequence, we will show that Paz's conjecture implies a result in the direction of Wielandt's conjecture, albeit with a worse order for the length.  

Let us conclude the introduction to this section by formally introducing the notion of Wie-length mentioned above.

\begin{defi}\label{def:wielength}
     Let $S$ be a finite subset of $M_n(\bbC)$. We say that $S$ is a \emph{Wie-generating system} if there exists $k \in \mathbb{N}$ with: 
 \begin{equation*}
     M_n(\bbC)=span \left\{A_1A_2\dots  A_k \, | \, A_i\in S \text{ for all }i\in[k] \right\} \, ,
 \end{equation*}
  and call the minimum such $k$ the \emph{Wie-length} of $S$, which we denote by $\wiel(S)$.
  \end{defi}

\subsection{Overview on Quantum Wielandt's Inequality}\label{subsec:Wielandt}

Similarly to the case of the length of an algebra, it is a natural, interesting question whether there exists a uniform upper bound on the Wie-length, which only depends on the dimension $n$ but not on the choice of a generating system for $M_n(\bbC)$. It was conjectured in \cite{perez2006matrix} that, in the worst-case scenario, it should be of order $O(n^2)$.

\begin{conj}[Quantum Wielandt's inequality]\label{conj:Wielandt}
  $\wiel(S)= O (n^2)$ for all (Wie-)generating systems $S\subseteq M_n(\bbC)$.
\end{conj}

The conjectured length cannot be improved in order, also due to the existence of generating pairs with a rank-one matrix, e.g. by the pair of matrices $A=\sum_{i=1}^n\ket{i+1}\bra{i}$ and $B=\ket{2}\bra{n}$. The first general result was that $\wiel(S)$ is of order $O(n^4)$ for all generating systems of $M_n(\bbC)$, and it followed from the existence of some words of length $O(n^2)$ with non-zero trace \cite{sanz2010quantum}. This upper bound was subsequently improved in \cite{michalek2019quantum} to $O(n^2\log n)$ with a similar technique as in \cite{shitov2019improved}, and this is the best general result to date. However, there is still a $\log n$-overlap compared to the conjectured optimal upper bound \cite{sanz2010quantum}.

Let us now further delve into the motivation for the study of this problem. Firstly, note that the name arises from its interpretation as a quantum version of Wielandt's inequality. We recall that Wielandt's inequality \cite{wielandt1950unzerlegbare} for classical channels quantifies the number of repetitions $k$ of a memory-less channel, denoted by a stochastic $n\times n$ matrix $W$, such that any probability distribution after transmission has non-zero probability for all events, i.e. all entries of $W^k$ are strictly positive. $W$ is primitive as long as the $k$ exists, and such a minimal $k$ is denoted as the \textit{index of primitivity} $p(W)$. Wielandt's inequality states that $p(W)$ is of order $O(n^2)$, as long as $W$ is primitive. In the quantum case mentioned above \cite{sanz2010quantum}, quantum channels are mathematically CPTP (complete positive and trace-preserving) maps, which can be expressed in terms of their Kraus operators, and thus in the quantum variant we are looking at the (Wie-)length of the Kraus operators to see after how many repetitions the quantum channel becomes strictly positive. 

There is a naturally related question, posed in \cite{rahaman2020wielandt}, where the authors looked more generally at positive (but not necessarily completely positive) maps, and proved that as long as a primitive positive map $\mathcal{E}$ acting on $M_n(\bbC)$ satisfies the \textit{Schwarz inequality}, the map becomes strictly positive in $O(n^2)$ iterations, meaning that they will send any positive-semidefinite matrix to a positive definite matrix in $O(n^2)$ steps. As a corollary, this bound gives a tighter upper bound on unital PPT and entanglement breaking channels than in \cite{perez2006matrix}.

\subsection{Relationship between Length and Wie-length}\label{subsec:Paz_implies_Wielandt}

Let us consider a finite subset $S$ of $M_n(\bbC) $. In this section, we compare the length  of $S$ (i.e. minimum length $l$ for words in elements of $S$ such that all words with length \underline{until} $l$ span  $M_n(\bbC) $) with its Wie-length (i.e. minimum length $l$ for words in elements of $S$ such that all words with length \underline{exactly} $l$ span  $M_n(\bbC) $). Note that the length and Wie-length of $M_n(\bbC) $ is the maximum over the length, resp. Wie-length, of all generating sets $S$. 

It is natural to expect a relation between length and Wie-length of $S$.  However, this is not at all straightforward, as both notions present some fundamental, important differences.  For example, it is clear that the length of $S$ is computed through an increasing chain of sets since the space spanned by words of length at most $k+1$ always contains the one spanned by words of length at most $k$. Moreover, when increasing the length of words by $1$, it increases the dimension of the space spanned at least by one unless the dimension reaches $n^2$.  The situation completely differs when considering $S$ to be a Wie-generating system, since now it is not guaranteed that the space spanned by increasing the length $k$ of words by $1$ is larger than the former one, even for small values of $k$. 

Nevertheless, there is some structure maintained. We first observe that, even though the Wie-length cannot be computed via an increasing chain of sets, the length actually stabilizes once they span the whole matrix algebra. In other words, if words of length $l$ span  $M_n(\bbC) $, then also words of length $m\geq l$ do, for any $m$.  This shows that the Wie-length is not only mathematically well-defined but also useful in computer-scientific problems.

\begin{lem}\label{lemma:wielandt_length_stabilizes} [Corollary of Lemma 5 in \cite{perez2010characterizing}]
   Let $S$ be a Wie-generating system of $M_n(\bbC)$ such that $\mathtt{Wie}\ell(S)= k$ for a certain $k \in \mathbb{N}$. Then, for every $m\geq k$, the following holds
   \begin{align}
     M_n(\bbC)=span \left\{A_1A_2\dots  A_m \, |  \, A_i\in S \text{ for all }i\in[m] \right\} \, .
 \end{align}
\end{lem}

Another important relation between these two notions is the equivalence for a generating system to generate $M_n(\bbC)$ in the traditional sense (i.e. with words until a certain length) and in the Wielandt sense (i.e. with words of exactly a certain length).

\begin{lem}\label{lem:paz_implies_wielandt}
 Let $S$ be a finite subset of $M_n(\bbC)$ which generates the matrix algebra $M_n(\bbC)$. It holds that
 \begin{align}
     \ell(S)\leq \mathtt{Wie}\ell(S)\leq (n^2+n)\ell(S)<\infty.
 \end{align}
\end{lem}

\begin{proof}
   The first inequality is trivial due to the definition of length and Wie-length. The second inequality is derived directly from Lemma 3.4 in \cite{michalek2019quantum}, which is inspired by Lemma 2 in \cite{sanz2010quantum}. The key idea here is to prove that a space spanned by products of a certain length includes all rank-one matrices. In \cite{michalek2019quantum}, one can constructively give an upper bound on the Wie-length using the rank and length $k$ of an existing non-nilpotent matrix which is a linear combination of some length-$k$ products. Clearly, there has to be a word in the basis that is non-nilpotent with length less than or equal to $\ell(S)$, as otherwise, this generating set would only generate traceless matrices, leading to a contradiction. Thus, one can take the upper bound on the length to be $\ell(S)$ and the upper bound on the rank to be the dimension $n$. This concludes the proof. 
\end{proof}

The previous lemma shows that if $S$ generates $M_n(\bbC)$, then it also Wie-generates it. Note that the converse implication is trivial, as words of a certain length $l$ are, in particular, of length $\leq l$. However, the quantitative relation obtained above between $\ell(S)$ and $\mathtt{Wie}\ell(S)$ is not expected to be optimal. This is due to the $n^2$ factor in front of $\ell (S)$, which would be expected to be of order $\Theta(n)$ instead. To illustrate this, there are available examples in the literature, such as generating pairs with one rank-one matrix \cite{GutermanLaffeyMarkovaSmigoc_PazConjecture_2018}, for which $\ell(S)=\Theta(n)$ and $\mathtt{Wie}\ell(S)=\Theta(n^2)$. Hence, in these cases, the natural conjecture would be:

\begin{conj}
   Let $S$ be a finite generating system of $M_n(\bbC)$. Then
 \begin{align}
     \ell(S)\leq \wiel(S)=O(n\,\ell(S)).
 \end{align}
\end{conj}

If this conjecture holds true, then Wielandt's conjecture could be immediately seen as a consequence of Paz's conjecture. This would provide an incredibly relevant connection between this classical problem in the field of linear algebra, and the more recent one that arises from applications to physics and information theory.

The opposite direction, namely a possible derivation of Paz's conjecture from Wielandt's conjecture, is not expected to hold. Imagine that there exists a general method to construct, from words of length $O(n^2)$ spanning $M_n(\bbC)$, a set of words of smaller length, bounded by  $O(n)$, so that they span $M_n(\bbC)$ as well. Then, this would in particular apply as well to the generic case. However, as we will show later in the text, in that case we have that $\Theta(\log n)$ is optimal for both length and Wie-length of generic generating systems. Therefore, this is a contradiction to the aforementioned ``reduction of length'' argument. This suggests that a possible approach to prove Paz's conjecture from a positive resolution of Wielandt's conjecture should separate the generic case and the worst-case scenario, in which we expect the length, resp. Wie-length, to scale with $\Theta(n)$, resp. $\Theta(n^2)$. However, determining whether each given generating system is in either of the two classes might be an even harder problem than solving the conjectures independently.

\section{A Generic Quantum Wielandt's Inequality}

As pointed out in \cref{conj:Wielandt}, the optimal general quantum Wielandt's inequality from the previous section, it is generally believed that the Wie-length of $M_n(\bbC)$ should be of order $O(n^2)$ for \textit{every} generating system. However, another relevant perspective to study related to that problem is what happens \textit{generically}, rather than in the worst-case scenario. More specifically, in this section we focus on a  \textit{probabilistic} study of the Wie-length, namely the Wie-length of a randomly chosen matrix subset of $M_n(\mathbb{C})$, with fixed cardinality, with respect to the Lebesgue measure. 

A simple counting argument shows that the Wie-length of any generating system $S$ scales as $\Omega(\log n)$, as we need at least $g^{ \wiel(S)} \geq n^2 $ words to generate $M_n(\bbC)$, with $g$ the cardinality of $S$. In the converse direction, sampling randomly from $M_n(\bbC)$, we observe that a computer ``almost always'' outputs a Wie-length of order $\Theta(\log n)$, i.e. the different words of length $\Theta(\log n)$ created from the generators are linearly independent in $M_n(\bbC)$ almost surely. It turns out that generating systems with longer lengths form a measure-zero set - it actually ``almost never'' happens for random matrices.  

It was first conjectured in \cite{perez2006matrix} that for a randomly chosen generating system, the Wie-length scales as $\Theta (\log n)$ and it almost never outputs the worst-case scenario. In \cite{cadarso2013wielandt}, the authors observed that the generating systems that do not span the whole matrix algebra with a certain length of matrix products form a projective algebraic subvariety of $\mathbb{C}^n\otimes \mathbb{C}^n\otimes\mathbb{C}^g$. Therefore, due to its irreducibility, it would be enough to find a single generating system spanning the whole matrix algebra with a certain length of matrix products to prove that this property can only fail with an exception of a zero measure set. Analytically, the bound was proven asymptotically for generic $S$ with more than 16 elements and large enough $n$ in \cite{cadarso2013wielandt} using quantum expanders \cite{Hastings2007quantumexpanders}. 

Nevertheless, an analytical, general proof of this numerical behavior was still missing. By carefully studying both the two previously mentioned directions, namely that of Paz's conjecture and Wielandt's conjecture, as well as by finding a tight relation between the notions of length and Wie-length, we have discovered that, in fact, \cite{KlepSpenko_SweepingWords_2016} already provides, in a different field, an almost complete proof of this feature for the generic case (up to a change of language regarding the objects involved). We rewrite that result in our setting below and call it the \textit{generic quantum Wielandt's inequality}. It consists of a bound $\Theta(\log n)$ for generic (Wie-)generating systems, and we include a sketch of the proof in our setting for completeness, leaving the most technical part to be consulted by the interested reader in \cite{KlepSpenko_SweepingWords_2016}.

\begin{thm}\label{gen.Paz}
  The set of generating systems $S$ of $M_n(\bbC)$ with $g$ elements that satisfy
  \begin{align}
      \mathtt{Wie}\ell(S)> 2\lceil \log_g n \rceil
  \end{align}
  has Lebesgue measure zero in $M_n(\bbC)^g$. In other words, given a generating system $S$ of $M_n(\bbC)$,
  \begin{equation}
       \mathtt{Wie}\ell(S) = O (\log n)
  \end{equation}
  almost surely.
\end{thm}

The key ingredient in the proof is that, by Zariski topology on Euclidean spaces, an argument for almost all generating systems follows by finding the existence of one generating system that has a minimal length to build up a basis. Furthermore, the possible words that span a basis provided by \cite{KlepSpenko_SweepingWords_2016} all have the same length. 

\begin{proof}
  
   We provide a sketch of the proof in Figure \ref{fig:sketch_proof}. Note that for simpler clarification, the figure and the proof below only concern the case in which $S$ contains $2$ elements, but it can be straightforwardly generalized to the cases with arbitrary $g\geq 2$. The roadmap works as follows:

\begin{figure*}
\includegraphics[scale=0.4]{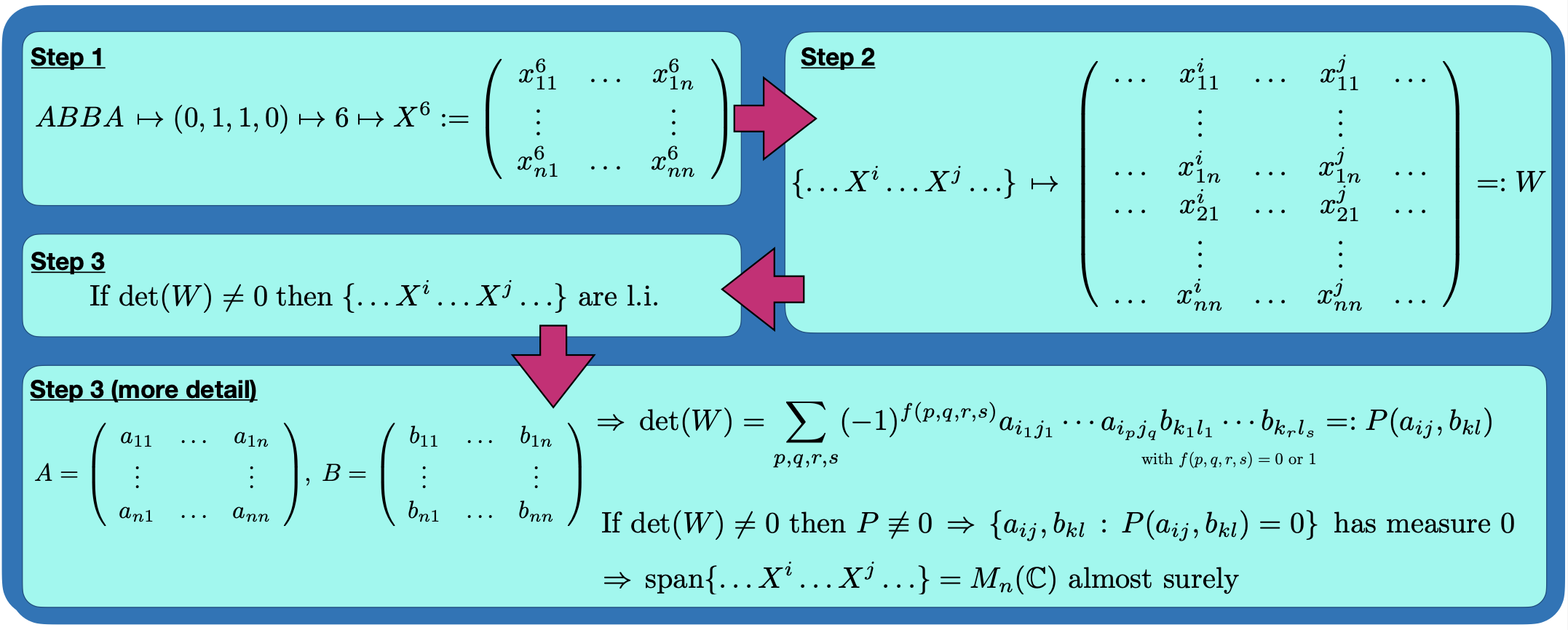}
	\caption{Sketch of the proof for Theorem \ref{gen.Paz}. Note that $1\leq i,j \leq 2^\ell$ in Step 2 and the number of elements in $\{ \ldots X^i \ldots X^j  \ldots\}$ is $n^2$. }
\label{fig:sketch_proof}
\end{figure*}

\vspace{0.2cm}

\textbf{\underline{Step 0. Select the words.}}

\vspace{0.2cm}

First, consider $n^2$ words of length $\ell$ in $A$ and $B$, namely products of the form 
$$\underbrace{ABBAB\ldots BA}_{\ell \text{ elements}} \, .$$ By the counting argument explained above, it is clear that $\ell = \Omega (\log n)$. Note that, in general, we will have $2^\ell >n^2$, and therefore we will not need all the words of length $\ell$ generated by $A$ and $B$, but just a selected subset of them with $n^2$ words.

\vspace{0.2cm}

\textbf{\underline{Step 1. Change notation of each word.}}

\vspace{0.2cm}

This is done to simplify the rest of the proof. Since we only consider two generators, we can rewrite each word in binary notation and identify each binary number with its decimal expression, as shown in Figure \ref{fig:sketch_proof} for the particular case of $ABBA$. In this way, we identify each word with a specific matrix and establish an order among them. 
\vspace{0.2cm}

\textbf{\underline{Step 2. Vectorize words and join them\\ in a matrix.}}

\vspace{0.2cm}

Each of the matrices in the previous step are of dimension $n \times n$. Thus, we can write the coordinates of each of them in a vector of $n^2 \times 1$ entries. We then write the $n^2$ vectors associated to the $n^2$ words in the columns of a matrix $W$ of dimension $n^2 \times n^2$ according to the order.

\vspace{0.2cm}

\textbf{\underline{Step 3. Compute the determinant of that matrix.}}

\vspace{0.2cm}

We now compute the determinant of $W$. Note that, if det$(W)\neq 0$, then all the words are linearly independent. More specifically, det$(W)$ is actually a polynomial of $2n^2$ variables, namely $\{ a_{ij}\}_{i,j=1}^n$ and $\{ b_{kl}\}_{k,l=1}^n$, the coefficients of $A$ and $B$ respectively. Therefore, if $P:=$det$(W)\neq 0$, then $P$ is not the identically-zero polynomial, and thus its zeroes have null Lebesgue measure. In other words, the set of words considered in Step 0 spans $M_n(\mathbb{C})$ almost surely. 

\vspace{0.2cm}

The remaining part to conclude the proof is to justify the existence of the words of Step 0. This is a highly non-trivial result shown in  \cite{KlepSpenko_SweepingWords_2016}, where the authors explicitly give $n^2$ words of length $2\lceil \log_g n \rceil$ such that $P$ is not the identically-zero polynomial. The key idea here is to prove the uniqueness of one monomial in $P$, implying $P$ is non-zero, by translating the question as decomposing the edges of a corresponding directed graph into paths with specific conditions.

\end{proof}

\section{Applications}\label{sec:applications}

The previous section has been devoted to reviewing the proof of quantum Wielandt's inequality in the generic case. Its main result, \cref{gen.Paz}, finds some very interesting applications in the field of quantum information theory. In this section, we clarify some of its implications in the contexts of injectivity of Matrix Product States and Projected Entangled Pair States, as well as for having eventually full Kraus rank (and for the primitivity index) of quantum channels.

\subsection{Matrix Product States (MPS)}\label{subsec:MPS}

Consider a pure quantum state $\ket{\psi} \in \mathbb{C}^{\otimes g^L}$ modelling a system of $L$ sites, each of which corresponds to a $g$-dimensional Hilbert space. If a translation-invariant pure state $\ket{\psi } $ can be written in the form
\begin{equation}\label{eq:MPS}
    \ket{\psi} = \underset{i_1, \ldots, i_L =1}{\overset{g}{\sum}} \text{tr} \left[ A_{i_1} \ldots A_{i_L}  \right] \ket{i_1 \ldots i_L}
\end{equation}
we say that $\ket{\psi}$ is a \textit{Matrix Product State} (MPS) with periodic boundary conditions, i.e. the system is placed on a ring. It admits the following standard graphical representation, where the legs joining the (square) matrices represent the contraction of coefficients given by matrix multiplication, and the vertical segments represent the physical index (see Figure \ref{fig:MPS}). For further information, we refer the reader to e.g. \cite{perez2006matrix,ReviewMPS}.

\begin{figure}[h!]
	\centering
	\includegraphics[scale=0.25]{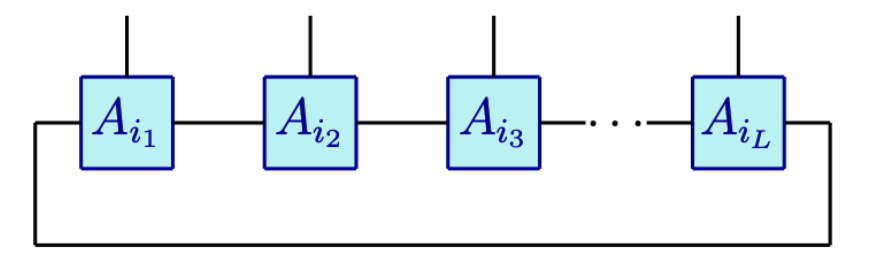}
	\caption{Graphical representation of a MPS.}\label{fig:MPS}
\end{figure}

For any $L \in \mathbb{N}$, let us consider the map $\Gamma_L : M_n(\mathbb{C}) \rightarrow \mathbb{C}^{\otimes g^L}$ given by
\begin{equation}\label{eq:Gamma_operator}
    \Gamma_L : X \mapsto \underset{i_1, \ldots, i_L =1}{\overset{g}{\sum}} \, \text{tr} \left[ X A_{i_1} \ldots A_{i_L} \right] \ket{i_1 \ldots i_L} \, .
\end{equation}
The study of the injectivity of this map is essential in the field of quantum spin chains, as it implies that the state represented by the MPS is the unique ground state of a certain gapped Hamiltonian. As mentioned in \cite{perez2006matrix}, it turns out that $\Gamma_L$ is injective if, and only if,
\begin{equation*}
   \text{span}\left\lbrace A_{i_1} \ldots A_{i_L} \, : \,1\leq i_1 , \ldots, i_L \leq g  \right\rbrace  = M_n(\mathbb{C}) \, , 
\end{equation*}
or, equivalently, 
\begin{equation*}
   \wiel(\lbrace A_1, \ldots  , A_g  \rbrace) \leq L \, .   
\end{equation*}

In \cite{perez2006matrix}, the authors introduced the following condition:

\vspace{0.2cm}

\textit{C1. There is a big enough $L\in \mathbb{N}$ such that the words of length $L$ span $M_n(\mathbb{C})$}.

\vspace{0.2cm}

One can check that, given a set of matrices $S=\{  A_1, \ldots,  A_g\}$, condition C1 holds for $S$ with probability $1$. In other words, any generic set of (at least two) matrices is a (Wie-)generating system. We consider here this \textit{generic} case. The natural question then is how $L$ scales with the dimension of the matrices as well as the number of generators. In \cite{perez2006matrix}, it was conjectured that $\Gamma_L$ should be injective for 
\begin{equation*}
    L \geq 2 \log_g n \, ,
\end{equation*}
and this was supported with numerical evidence as well as an analytical proof for $g=n=2$. In the subsequent paper \cite{cadarso2013wielandt}, this result was extended to the following inequality
\begin{equation*}
    L \geq \left[ 8 \log_g n \right] +1   \, ,
\end{equation*}
which was shown to be valid for $g>16$ and $n$ large enough, using the technique of quantum expanders introduced in  \cite{Hastings2007quantumexpanders}. For small values of $g$ and $n$, their result was complemented with numerical simulations.

In this direction, we can complement the previous findings by proving a similar bound for $L$, with a better multiplicative factor, which is valid for every $g$ and every $n$. Indeed, our main result from \cref{gen.Paz} translates as follows:

\begin{cor}\label{cor:generic_L_MPS}
 Given $L \in \mathbb{N}$ such that
 \begin{equation}
     L \geq 2 \lceil \log_g n  \rceil  \, ,
 \end{equation}
 the map $\Gamma_L$ is injective almost surely with respect to Lebesgue measure.
\end{cor}

A particular case of this result is Fact 1.2 of \cite{Lancien2022randomPEPS} for MPS, where the authors showed injectivity of their random MPS almost surely whenever  $g \geq n^2$. 

In particular, as mentioned above, the fact that $\Gamma_L$ is injective implies that the MPS is the unique ground state of a local Hamiltonian (constructed from some orthogonal projections onto some subspaces), which is called the \textit{parent Hamiltonian} of the MPS. In this context, our result  \cref{gen.Paz} has the following implication:

\begin{cor}\label{cor:generic_L_MPS_Interpret}

A random translation-invariant MPS with periodic boundary conditions on $M \geq L$ sites, with physical dimension $g$ and bond dimension $n$, is with probability 1 the unique ground state of an $(L + 1)$-local Hamiltonian, where $L = 2 \lceil \log_g (n)\rceil$. 
\end{cor}

Let us conclude this subsection by discussing the most general possible case for this problem. For such a case, Conjecture \ref{conj:Wielandt} can be rephrased in the current setting as follows:

\begin{conj}[Conjecture on the injectivity of $\Gamma_L$]\label{conj:injectivityGammaL}
  There exists a function $L(n)=\Theta(n^2)$ such that for any set of matrices $S= \{A_1 , \ldots , A_g \} \subseteq M_n(\bbC) $ satisfying condition $C1$, the map $\Gamma_L$ is injective for any $L\geq L(n)$. 
\end{conj}   

As mentioned above, the previous conjecture remains still an open problem, and the best result to date contains a $\log n$ overhead, i.e. there is $L(n)=O(n^2 \text{log } n )$ such that $\Gamma_L$ is injective for any $L\geq L(n)$ \cite{michalek2019quantum}.

\subsection{Projected Entangled Pair States (PEPS)}\label{subsec:PEPS}

The previous section has been devoted to proving that almost all Matrix Product States, i.e. rank-3 tensors $\calA\in  \mathbb{C}^g \otimes  (\mathbb{C}^n)^{\otimes 2}$, on a ring with more than $L=2\lceil \log_g n \rceil$ sites constitute an injective region for the map $\Gamma_L$. It is natural to consider whether this result can be extended to larger dimensions. 

Let us focus for the moment on dimension $2$. An extension of the previous method would be to construct states on a 2-dimensional finite square lattice (i.e. ${R}\subset \mathbb{Z}^2$ of size $L_1 \times L_2$) with periodic boundary condition (that is, to construct tensor network states on the torus). Such states are called \textit{PEPS (Projected Entangled Pair States)} \cite{perez2008PEPS}. Again, we assume that the tensors assigned to each vertex are site-independent, so the PEPS is uniquely given by the size of the lattice and a rank-5 tensor $\calA\in\mathbb{C}^g \otimes  (\mathbb{C}^n)^{\otimes 4}$, where $g$ denotes the physical dimension and $n$ denotes the bond dimension. It is not hard to see that PEPS also fulfil an area law, and for each PEPS one can construct a so-called parent Hamiltonian which has the PEPS as a ground state. However, similarly as in the case of MPS, in order to guarantee that the given PEPS state is the unique exact ground state of a gapped local acting Hamiltonian, we have to ensure injectivity.

More specifically, a tensor $\calA$ is called \textit{normal} if there exists a region $ R$ such that the linear map $\Gamma_R: (\mathbb{C}^n)^{\otimes |\partial R|}\rightarrow (\mathbb{C}^g)^{\otimes |R|}$, which maps boundary conditions living in the virtual space to vectors in the physical Hilbert space, is injective (cf. Figure \ref{fig:PEPS}). 

\begin{figure}[h!]
	\centering
	    \includegraphics[scale=0.25]{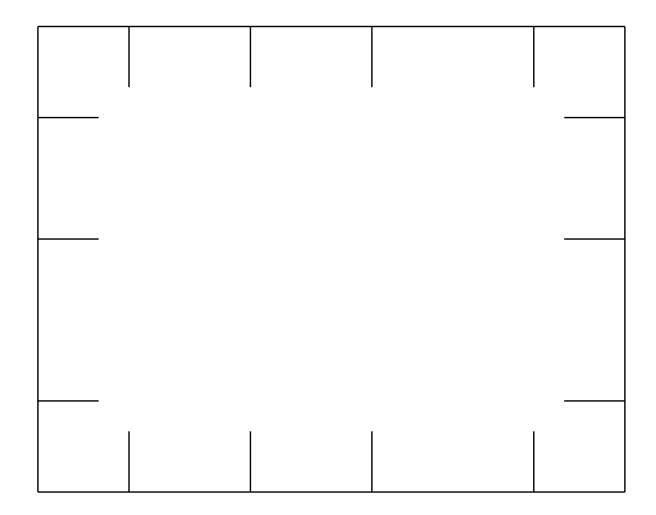}
    \includegraphics[scale=0.25]{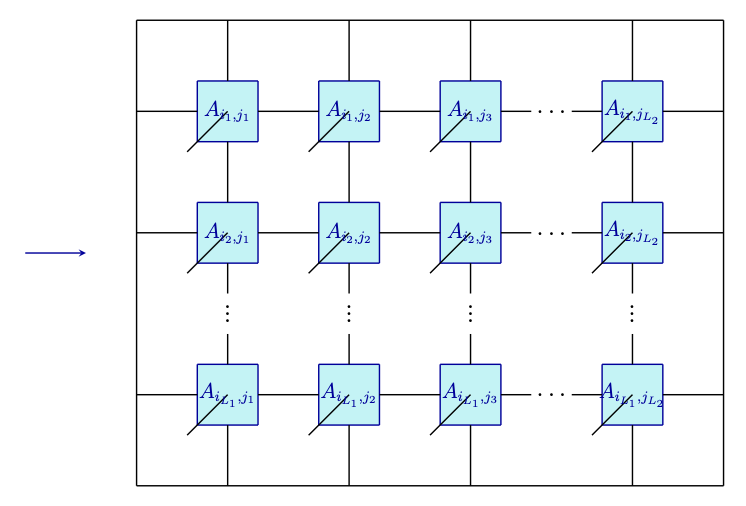}
	\caption{Graphical representation of the map $\Gamma_R$ for PEPS.}\label{fig:PEPS}
\end{figure}

Assuming for simplicity $L=L_1=L_2$, we require $\Gamma_L: (\mathbb{C}^n)^{\otimes 4L}\rightarrow (\mathbb{C}^g)^{\otimes L^2}$ to be injective for a large enough $L\in \mathbb{N}$, and the minimal such $L$ is called the \textit{index of injectivity} of the PEPS. The goal is then, when assuming $\calA$ to be normal, to provide a general $L_0\in\mathbb{N}$ only depending on the physical and virtual dimension $g$ and $n$, resp., such that $\Gamma_L$ is injective. Such a question can be interpreted as \textit{a tensor version of quantum Wielandt's inequality} \cite{michalek2019tensor}. In fact, this is one of the most relevant mathematical open questions for PEPS, as described in \cite{Cirac_2019}, namely whether one can find a general, optimal upper bound for the injectivity index in terms of the bond dimension of a normal tensor.

We note that, in higher dimension, the  index of injectivity is still reasonably defined, because once the region is injective, any region including it is also injective \cite{michalek2019tensor}. However, in the $1$-dimensional case, we take advantage of writing the tensor as a tuple of matrices and the question can be simplified just by considering the span of products of matrices. For higher dimensions, the structure of PEPS complicates the contraction: a general PEPS is computationally intractable since contracting (translation-invariant) PEPS is a $\sharp$P-complete problem \cite{schuch2007computational,scarpa2020projected} and also in average-case hard \cite{haferkamp2020contracting}. The only known fact about the tensor version of quantum Wielandt's inequality to date is the existence of a function $f(n)$ that bounds the index of injectivity $i(\calA)$ for any normal $A$ with bond dimension $n$ \cite{michalek2019tensor}. It is even unknown in general whether $f(n)$ can be a polynomial function or whether it is computable.

Here, we focus again on the generic case. A key observation is that, for a given $L$, the tuples $\{A_1, \dots A_g \}$ with $A_i \in (\mathbb{C} ^{n})^{\otimes 4}$ for which $R$ is not an injective region form a Zariski closed set of $((\mathbb{C} ^{n})^{\otimes 4} )^g$ (cf. \cite{michalek2019tensor}, at the end of the proof of Theorem 3.4). That is, if there exists a tensor $\calA_0$ such that $R$ is an injective region, $R$ is an injective region for almost all $\calA\in\mathbb{C}^g \otimes  (\mathbb{C}^n)^{\otimes 4}$ (explained in detail in the proof of \cref{cor:generic_L_PEPS}).  Due to that, we can construct an upper bound for the generic index of injectivity by proving ``existence'' of a normal tensor for which the degeneracy does not increase too much while $L$ increases. A simple generalization of the 1D case already provides a good upper bound with $L=\Theta(\log n)$ on the generic case. We look at a special type of string-bond states \cite{schuch2007strings}.

\begin{thm}\label{cor:generic_L_PEPS}
 Consider a finite region $R\subset \mathbb{Z}^2$ of side length $L \in \mathbb{N}$ and $\ket{\psi}$ a PEPS which is translation invariant and with periodic boundary conditions as above. Suppose $g\geq 4$ denotes the physical dimension and $n$ refers to the bond dimension. If $L$ is such that
 \begin{equation}
     L \geq  2\lceil \log_{\lfloor g ^{1/2}\rfloor} n  \rceil  \, ,
 \end{equation}
 the map $\Gamma_L$ is injective almost surely. Additionally, $\ket{\psi}$ is the unique ground state of a local Hamiltonian with probability $1$.
\end{thm}

\begin{proof}
    Let us first consider the special case that $g$ is a squared number, i.e. $g=d^2$ for some $d \in\mathbb{N}$. Suppose $\{B_1,\dots B_d\}$ and $\{\tilde{B}_1,\dots, \tilde{B}_d\}$ are two generating systems of $M_n(\bbC)$ with Wie-length $2\lceil \log_d n\rceil$ (cf. Theorem \ref{gen.Paz}).

    In this case, we construct a $g$-tuple $\{A_1,\dots A_g\}$, for which we assume that each of them is of the form 
    $A_k=A_k^v \otimes A_k^h$, with $v$ denoting vertical and $h$ horizontal. Moreover, $A_k^v = B_{i_k}$ and $A_k^h =\tilde{B}_{j_k}$ with $k=i_k(n-1)+j_k$. Then, we can contract first vertically
and subsequently horizontally (see Figure \ref{fig:string-bond-PEPS}).

 \begin{figure}[h!]
     \centering
     \includegraphics[scale=0.25]{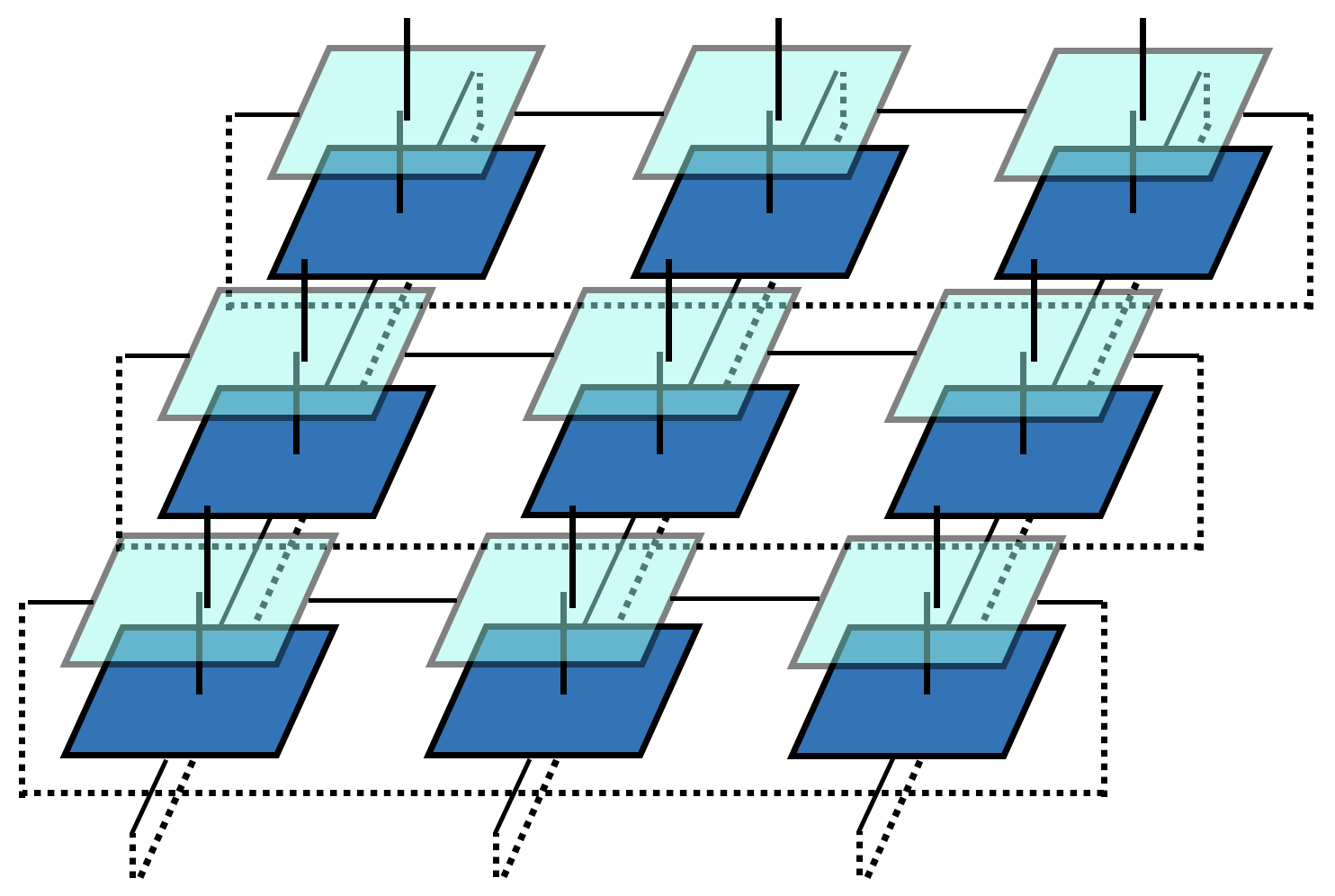}
     \caption{Graphical representation of a string-bond state.}
     \label{fig:string-bond-PEPS}
 \end{figure}   

The PEPS we obtain is the following:
\begin{align*}
    \ket{\varphi_\calA}& =
    \sum_{i_1,\dots i_{L}, j_1,\dots ,j_{L}=1}^g \prod_{i=i_1,\dots, i_{L}} \Tr [A_{i,j_1}^h\dots A_{i,j_{L}}^h] 
    \\ & \prod_{j=j_1,\dots ,j_{L}} \Tr [A_{i_1,j}^v\dots A_{i_{L},j}^v] \ket{i_1,j_1}\dots \ket{i_{L}, j_{L}} \\
    & = \underset{\tiny \begin{array}{c}
         i_1... i_{L} =1 \\
          j_1...j_{L} =1
    \end{array}}{\overset{g}{\sum}}\varphi^{\mathcal{A}}_{j_1i_1,\dots, j_Li_L} \ket{i_1,j_1}\dots \ket{i_{L}, j_{L}} \, ,
\end{align*}
where the coefficients are given by
\begin{align*}
    & \varphi^{\mathcal{A}}_{j_1i_1,\dots, j_Li_L}=\Tr[\bigotimes_{j=j_1\dots j_L}A_{j,i_1}^h\dots A_{j,i_L}^h\bigotimes_{i=i_1\dots i_L}A_{j_1,i}^v\dots A_{j_L,i}^v] \, 
\end{align*}
and our goal is to show that the $g^{L^2}$ tensors 
\begin{align*}
\bigotimes_{j=j_1,\dots, j_L}A_{j,i_1}^h\dots A_{j,i_L}^h\bigotimes_{i=i_1,\dots ,i_L}A_{j_1,i}^v\dots A_{j_L,i}^v \, ,
\end{align*}
where we identify the pairs in the following form $\{(j_1,i_1),(j_2,i_1), \dots ,(j_L,i_L) \equiv 1,\dots,g \equiv (1,1),\dots ,(d,d)\}$, 
span $(\mathbb{C}^n)^{\otimes 4L}$. Due to the special choice of the tensors, it is equivalent to prove that the terms
\begin{align*}
\bigotimes_{j=j_1,\dots, j_L}B_{i_1}\dots B_{i_L}\bigotimes_{i=i_1,\dots ,i_L}\tilde{B}_{j_1}\dots \tilde{B}_{j_L}
\end{align*}
with $\{j_1,i_1,j_2,i_1, \dots ,j_L,i_L=1,\dots,d\}$
span $(\mathbb{C}^n)^{\otimes 4L}$. Next, choose $L:=2\lceil \log_d n\rceil$. We can find $n^{4L}$ elements for a basis using the following procedure: For $j_1$ there exist $n^2$ elements in the set $\{B_{i_1}\cdots B_{i_L} |\, i_1,\dots, i_L=1,\dots,d\}$ forming a basis for $M_n(\bbC)$, and the same happens for each $j_2$ to $j_L$ and $i_1,\dots ,i_L$. Tensorizing the basis $2L$ times yields basis for $(\mathbb{C}^n)^{\otimes 4L}$.

In general, for arbitrary $g\geq 2^2$ we can take $d=\lfloor g^{1/2}\rfloor$ and choose the rest of the generators randomly. The space generated using 
this subset of $d^2$ generators is not larger than the one generated by $g$. Thus we still ensure that there exists a tensor such that, with 
\begin{align*}
    L\geq 2 \lceil \log _{\lfloor g^{1/2 }\rfloor} n\rceil \, ,
\end{align*}
we obtain already injectivity for the PEPS.

For a fixed $L\times L$ grid (or a more general region), the set $\mathcal{V}_L(\calA)$ containing the tensors $\calA$ for which the region is not injective, corresponds to the intersection of zero sets of some polynomials in the entries of $\calA$. More specifically, suppose $L\geq 2 \lceil \log _{\lfloor g^{1/2 }\rfloor} n\rceil$. We first notice that determining whether the generated tensors in $(\mathbb{C}^n)^{\otimes 4L}$ include a basis can be equivalently identified as deciding whether a $n^{4L}\times g^{L^2}$ matrix, whose columns consist of the vectorized tensors, has full rank. For our specific choice $\mathcal{A}_0=\sum_{k=1}^g A_k\ket{k}$, there exist $n^{4L}$ generated tensors that are linearly independent. We fix the choice of the words, and this time we look at the determinant of a $n^{4L}\times n^{4L}$ matrix by choosing these $n^{4L}$ elements out of the $g^{L^2}$ tensors. Then the determinant of this matrix is a polynomial over the entries of a tensor $\calA\in (\mathbb{C}^n)^{\otimes 4}\otimes \mathbb{C}^g$, abbreviated as $p(\calA)$. Since $p(\calA_0)\neq 0$, this immediately implies that the polynomial is not a zero polynomial. Moreover, $\mathcal{V}_L(\calA)$ is contained in the zero set of this non-zero polynomial (otherwise the $L\times L$ region would be injective for this tensor), which has null Lebesgue measure. 
As a consequence, $\mathcal{V}_L(\calA)$ has null Lebesgue measure, so we conclude that for almost all translation invariant PEPS with periodic boundary condition, the square grid with 
\begin{align*}
     L\geq 2 \lceil \log _{\lfloor g^{1/2 }\rfloor} n\rceil
\end{align*}
is an injective region.

\end{proof}

The argument 
    can be extended to any high-dimensional PEPS on squared lattices by a similar construction. In general, the tensor version of quantum Wielandt's inequality is defined on an $m$-dimensional square grid $\subset \mathbb{Z}^m$ for any $m\in \mathbb{N}$ in \cite{michalek2019tensor}. We do not use at any point that we work on a 2-dimensional lattice, except for the construction of tensor $\calA_0$. But in $m$-dimension, we can pick the tensor $\calA_0$ described by the tuple $\{A_1,\dots,A_g\}$, with the first $(\lfloor g^{\frac{1}{m}}\rfloor)^m$ terms $A_i\in (\mathbb{C}^n)^{\otimes 2m}$ fulfilling the form $\bigotimes_{k=1,\dots ,m}B_{k(i)}^k$ and $k(i)\in\{1,\dots, \lfloor g^{\frac{1}{m}}\rfloor\}$. 

\begin{cor}
    Consider a finite region $R\subset \mathbb{Z}^m$ of side length $L \in \mathbb{N}$ and $\ket{\psi}$ a PEPS which is translation invariant and with periodic boundary conditions as above. Suppose $g\geq 2^m$ denotes the physical dimension and $n$ refers to the bond dimension. Then for almost all $\calA\in (\mathbb{C}^n)^{\otimes 2m}\otimes \mathbb{C}^{g}$, the index of injectivity of $\calA$ satisfies 
    \begin{equation}
        i(\calA)\leq 2\lceil \log_{\lfloor g ^{1/m}\rfloor} n  \rceil  \,. 
    \end{equation}

\end{cor}

\begin{rem}
   For an $m$-dimensional square grid, 
   \begin{align*}
       \Gamma_{L,m}: (\mathbb{C}^n)^{\otimes 2mL^{m-1}}\rightarrow (\mathbb{C}^g)^{\otimes L^m} \, ,
   \end{align*}
    being injective implies that $i(\calA)\geq 2m\log _g n$ for any tensor $\calA$ (and the corresponding PEPS on the grid). Thus, our result is close to being optimal for the generic case. In particular, we observe that no matter how the dimension $m$ changes, in the generic case $i(\calA)=\Theta(\log n)$, namely ``the best case we can expect is already the average case''. However, the worst case for $m=2$ could potentially be enormous, since there does not exist yet a polynomial uniform upper bound on $i(\calA)$.
\end{rem}

\subsection{Primitivity Index of Quantum Channels}\label{subsec:capacities}

Consider a quantum channel $\mathcal{E}$, i.e. a completely positive trace-preserving linear map. We say that the channel is \textit{primitive} if there is an integer $\ell \in \mathbb{N}$ such that, for any positive semi-definite matrix $\rho$, the $\ell$-fold application of the channel to $\rho$ is positive definite; namely, if $\mathcal{E}^\ell(\rho) >0$ for every $\rho \geq 0$.  The minimum $\ell$ for which this condition is fulfilled is called \textit{index of primitivity} and is denoted by $q(\mathcal{E})$. 

There is a one-to-one correspondence between a quantum channel $\mathcal{E}$ and its Choi matrix $\omega(\mathcal{E})=( \operatorname{id}\otimes \mathcal{E})(\Omega)$ with $\Omega = \sum_{i,j=1}^n \ket{ii}\bra{jj}$. The rank of $\omega (\mathcal{E})$ is the \textit{Kraus rank} of the channel. The channel is said to have eventually full Kraus rank, if there exists some $n\in\bbN$ such that $\mathcal{E}^n$ has full Kraus rank, and the minimal such $n$ is denoted by $i(\calA)$, for $\calA=\{A_i\}_{i=1}^g$ the Kraus operators of $\mathcal{E}$. In \cite{sanz2010quantum}, it is shown that primitivity for a quantum channel is equivalent to eventually having full Kraus rank, with the index $i(\calA)$ lower bounded by the index of primitivity $q(\mathcal{E})$, and to being strongly irreducible, i.e. the channel having a unique eigenvalue of modulus $1$ whose corresponding eigenvector is positive definite. Differently from this, in the classical case all these notions, as well as the associated indices, coincide.  

It is not difficult to see that the index of eventual full Kraus rank for a quantum channel is equivalent to the Wie-length for its Kraus operators. Indeed, note that if $\mathcal{E}$ has Kraus operators $\{A_i \}_{i=1}^g$, i.e.
\begin{equation*}
    \mathcal{E} (X) = \underset{i=1}{\overset{g}{\sum}} A_i X A_i^\dagger \, ,
\end{equation*}
then $\mathcal{E}^m$ having full Kraus rank is equivalent to the condition that the set
\begin{equation*}
\text{span}\{ X_{1} \ldots X_m \, | \, X_i = A_j \, \text{ for } i \in [m] , j \in [g]  \} = M_n (\mathbb{C}  )   ,
\end{equation*}
 or, equivalently, $ \mathtt{Wie}\ell\left( \{ A_1, \ldots , A_g  \}\right)= i(\calA)$. 

In \cite{sanz2010quantum}, the authors show that, in general, the index of eventual full Kraus rank for a channel (and thus the index of primitivity) is of order $O(n^4)$. 
More specifically, they present some specific channels with $i(\calA)$ of order  ${O}(n^2)$, the best expected general order. Therefore, an implication of \Cref{gen.Paz} is that it notably improves such an order for the Kraus rank in the generic case.

\begin{cor}\label{cor:generic_Kraus_quantum_channels}
 Given a random quantum channel $\mathcal{E}: M_n (\mathbb{C}) \rightarrow M_n (\mathbb{C}) $ with Kraus operators $\calA=\{ A_1, \ldots, A_g  \} \subseteq M_n(\bbC)^g$ under the condition 
        $\sum_{i=1}^{g} A_i^\dagger A_i=\1$, the minimal repetitions $i(\calA)$ so that the channel eventually has full Kraus rank (and thus its index of primitivity $q(\mathcal{E})$) is of order $\Theta(\log n)$. 
\end{cor}

The implication follows almost directly from \Cref{gen.Paz} and \cite{sanz2010quantum}, with the only exception being that the trace-preserving condition of a channel requires the Kraus operators to satisfy specific conditions. Therefore, one needs to adjust the proof to additionally state that, when considering the generating systems $S=\{A_1, \dots ,A_g\}\subseteq M_n(\bbC)$ with $g$ elements under the normalization condition $\sum_{i=1}^{g} A_i^\dagger A_i=\1$, the elements with length $\wiel(S)> 2\lceil\log_g n\rceil$ have measure zero with respect to the natural Lebesgue measure restricted to this subset.  A proof for this was pointed out by Aubrun and Jingbai and is included in the Appendix for convenience. Note that it holds with probability $1$.  Thus, even though it does not cover all quantum channels, it covers a subset of full measure. In particular, for experimental applications with random quantum channels, it would always be the case that the index of primitivity is $\Theta(\log n)$ as opposed to  $O(n^2)$.

As in the previous subsection, the best bound for the index of eventual full Kraus rank of any primitive quantum channel to date is of order ${O}(n^2 \log n)$ \cite{michalek2019quantum} and the conjectured bound (which has been recently proven at least for the index of primitivity of positive maps which also satisfy the Schwarz inequality, as shown in \cite{rahaman2020wielandt}) is as follows:

\begin{conj}[Conjecture on the injectivity of $\Gamma_L$]\label{conj:injectivityChannel}
  Let $\mathcal{E}: M_n (\mathbb{C}) \rightarrow M_n (\mathbb{C}) $ be a primitive quantum channel with Kraus operators $\calA= \{ A_1, \ldots , A_g  \}$. Then, the index of eventually having full Kraus rank $i(\calA)$ and its primitivity index $q(\mathcal{E})$ are of order $O(n^2)$. 
\end{conj}   

To conclude this section, we explore the implications of our result in the context of the zero-error capacity $C_0$ of a noisy quantum channel. This is given by the largest quantity such that there is a sequence of increasing-length blocks with rates of transmission approaching it and the probability of error after decoding is $0$. In \cite{sanz2010quantum}, they show a dichotomy behaviour for the zero-error capacity of a quantum channel $\mathcal{E}$ with a full-rank fixed point. As a consequence of (non-)primitivity, either of the two following options holds:
\begin{itemize}
    \item $C_0^n(\mathcal{E}) \geq 1$ for all $n \in \mathbb{N}.$
    \item $C_0^{q(\mathcal{E})}(\mathcal{E}) =0.$
\end{itemize}
By Corollary \ref{cor:generic_Kraus_quantum_channels}, we have that, with probability $1$, $q(\mathcal{E})$ can be taken to be of order $\Theta(\log n)$ in the previous dichotomy.

\section{The Length of a Lie Algebra}

Until now, we discussed that a subset of $M_n(\bbC)$ generates the whole matrix algebra and has a basis of words of length $O(\log n)$ almost surely. Moreover, the previous length matches the asymptotic lower bound provided by the free algebra, as shown by a simple counting argument. This phenomenon leads to the guess that different words with small length are actually linearly independent in $M_n(\bbC)$ with probability 1 unless the dimension saturates.

The purpose of this section is to study the analogous problem for Lie algebras, instead of finite-dimensional associative algebras. Here, we will introduce the Lie-length of a Lie algebra, describe the problem of finding the optimal Lie length, conjecture the optimal solution, provide some numerical evidence to support it in the generic case, and mention some possible applications of such an analytical result, which is left to study in future work.

\subsection{Preliminaries and Notations}

Recall that a \textit{Lie algebra} $(\calA, [\cdot,\cdot])$ is an algebra with an operation $[\cdot,\cdot]$ satisfying skew-symmetry and Jacobi identity, i.e. for all $a,b,c$ in $\calA$,
\begin{align*}
    &[a,b]+[b,a]=0,\, \\ &[a,[b,c]]+[b,[c,a]]+[c,[a,b]]=0. 
\end{align*}

Given a subset $U$ of a Lie algebra $\calA$, the \textit{Lie subalgebra} generated by $U$ is the smallest Lie subalgebra of $\calA$ containing $U$, denoted by $\ttLie\{U\}$, and $U$ is called a \textit{Lie-generating system} of $\calA$ if $\ttLie\{U\}=\calA$. A natural question that arises for finite-dimensional Lie algebras is the minimum required length $\ell \in\bbN$ such that brackets of length upper bounded by $\ell$ span the whole Lie algebra. We will call this $\ell$ the \textit{Lie-length} of the Lie algebra.

The first issue we might encounter is whether the notion of Lie-length is well-defined, since we lack associativity in this context. Nevertheless, we see that this is actually not a problem, as a consequence of  skew-symmetry and Jacobi-identity:

\begin{lem}[Nested brackets, \cite{bonfiglioli2011topics}, Theorem 2.15]\label{nestedbrackets}
	Let $\calA$ be a Lie algebra and $U\subseteq \calA$. Set
	\begin{align*}
	U_1&=\text{span}\{U\};\, U_n=\text{span} [U,U_{n-1}]\,  , \, n\geq 2.
	\end{align*}
	Then $\ttLie\{U\}=\text{span} \{U_n, n\in \mathbb{N}\}$, and 
	\begin{align*}
	[U_i,U_j]\subseteq  U_{i+j},\quad \text{    for all } i,j \in \mathbb{N}.
	\end{align*}
\end{lem}
All brackets thus are linear combinations of right-nested commutators (i.e. commutators of the form $[u_1,[u_2, \dots[u_{k-1},u_k]]]]$) of the same length.

Therefore, we can define the \textit{Lie-length} of a Lie generating system $U$ as
following.

\begin{defi}
 The \textit{Lie-length} of a Lie-generating system $U$ of a finite-dimensional Lie-algebra $\calA$ is the minimal length of the (right-nested) brackets such that the brackets span the whole Lie-algebra, namely
    \begin{align}
     \ttLie\ell(U)=\min \{ \ell | \calA=\text{span} \{U_n, n\leq \ell\}\} , 
 \end{align}
 where  $U_n$  is defined as:
\begin{align*}
	U_1&=\text{span}\{U\};\, U_n=\text{span} [U,U_{n-1}]\,  , \, n\geq 2.
	\end{align*}
\end{defi}

 \subsection{Overview on Lie-length of Matrix Lie Algebras}

As in the case of matrix algebras, we are interested in the study of the Lie-length of certain Lie algebras. In particular, we focus here on  $\fraksu(n)$. As far as we know, there is no general upper bound for any Lie-generating system in the literature, and even if it existed, it could not be better than $O(n \log n)$, since the commutators are themselves linear combination of words, so no general result could outperform the one for associative algebras.

Regarding the lower bound, Witt's formula \cite{serre2009lie, reutenauer2003free} provides the maximal number of linearly independent length-$k$ brackets by a counting argument, performed by only taking into consideration the basic properties of the Lie algebra, namely skew-symmetry and Jacobi identity.  Its formal version appears below.

\begin{thm}[Witt, \cite{serre2009lie}, Theorem 4.2]
	Suppose $X$ has cardinality $g$, and $a_k$ is the dimension of the homogeneous part of degree $k$ of the free Lie algebra over $X$. Then, 
	\begin{align}
	a_k=\frac{1}{k}\sum_{d|k} \mu(d)g^{k/d} \, ,
	\end{align}
	where $\mu: \bbN_+\rightarrow \{-1,0,1\}$ denotes the Möbius function
	\begin{align}
	    \mu(d)= \sum_{\substack{1\leq k \leq d \\ \text{gcd}(k,d)=1}} e^{2\pi i \frac{k}{d}} \, .
	\end{align}
\end{thm}
In the particular case of a matrix Lie algebra of dimension $\Theta(n^2)$, this shows an asymptotic lower bound of order $\Omega(\log n)$ on the Lie-length of a generating set containing $n\times n$ matrices \cite{lloyd2019efficient}. This lower bound matches perfectly the generic case of the matrix algebra setting and the numerical result obtained for Lie algebras that we will explain below. 

 \subsection{A Generic Upper Bound on Lie-length of Matrix Lie Algebras}
 
  As we can construct a basis for the Lie algebra by restricting to right-nested brackets, we could search for a basis with minimal length through a tree structure algorithm, like the ``Lie-Tree'' algorithm from \cite{elliott2009bilinear}, which does a breadth-first search and computes the Lie-length as follows: 
  \begin{itemize}
      \item  At each step, the length increases by one and we compute a new set of right-nested commutators.
      \item We consider one of them, evaluate it as a matrix, and discard it if it is linearly dependent on the previous matrices. 
      \item We repeat this with all the new right-nested commutators. 
  \end{itemize}
  The algorithm stops when there are enough basis elements or the length reaches the dimension itself (which means that the set does not generate the Lie algebra). In comparison to the numerical experiment for matrix algebras, at each step, we expect linearly dependent terms for any matrix Lie algebra due to the Jacobi identity and skew-symmetry. 

\begin{figure}[h!]
	\centering
	\includegraphics[scale=0.25]{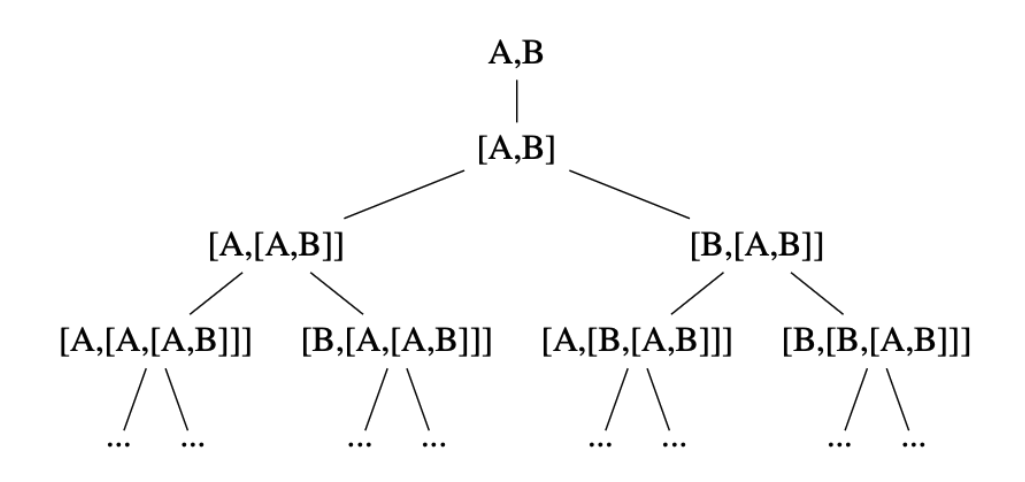}
	\caption{Breadth-first search along the Lie Tree for a Lie-generating pair. The obvious linearly dependent terms in the first two layers are omitted.}\label{Lietreegraph}
\end{figure}

When testing this algorithm for random pairs in $\fraksu(n)$ for $n\leq 20$ (see Figure \ref{graph:su(n)}), we observe that the Lie-length scales as $\Theta(\log n)$ and it does not change when we randomly choose another initial pair. Moreover, Witt's formula \cite{serre2009lie,reutenauer2003free} shows that $\Omega(\log n)$ is a lower bound on the Lie-length. Let us remark that similar numerical results with the same asymptotic behavior could be obtained for $gl(n,\bbR)$, $gl(n,\bbC)$, $\mathfrak{o}(n)$, $\mathfrak{u}(n)$, $\mathfrak{so}(n)$ as well. As they are completely analogous, we omit the details for other matrix Lie algebras and focus on $\fraksu(n)$ hereafter.

\begin{figure}[h!]
	\centering
	\includegraphics[scale=0.5]{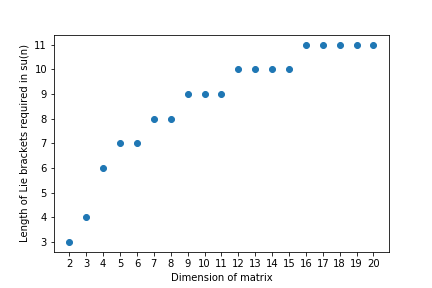}
	\caption{The length required to enable random pairs of traceless skew-Hermitian $n\times n$ complex matrices to generate $\fraksu(n)$, which scales as $\Theta(\log n)$.}\label{graph:su(n)}
\end{figure}

The aforementioned numerical result supports the conjecture that generically the length of a Lie-generating system should match at least asymptotically the bound given by the counting argument, or, in other words, that all linearly independent Lie brackets almost surely output linearly independent matrices, as long as the dimension has not saturated.

\begin{conj}\label{conj:lielength}
   Let $S$ be a random Lie-generating set of $\fraksu(n)$, then 
   \begin{align}
       \ttLie\ell(S)=\Theta(\log n)\,\,\text{almost surely.}
   \end{align}
\end{conj}

An analytical proof for this conjecture, though, still remains open. Nevertheless, we can show that most generating systems (with the same number of elements) have the same Lie-length, i.e. most generating systems are generic. 

\begin{lem}
    The elements in $\fraksu(n)^g$ that are not in $\underset{S\in \fraksu(n)^g}{\mathrm{argmin}}\,\mathtt{Lie}\ell(S)$ have measure zero in $\fraksu(n)^g$.
\end{lem}

\begin{proof}
   Let us denote
   \begin{align*}
       l:= \underset{S\in \fraksu(n)^g}{\mathrm{min}}\,\mathtt{Lie}\ell(S).
   \end{align*} 
   The minimum is attained since the length takes value in $\bbN$ and has a non-zero lower bound. Next, $\fraksu(n)$ can be regarded as an ($n^2-1$)-dimensional $\bbR$-vector space with basis
   \begin{align*}
       \{i(E_{hh}-E_{h+1,h+1}), E_{jk}-E_{kj}, i(E_{jk}+E_{kj})\}
   \end{align*}  
   with $h=1,\dots,n-1$ and $1\leq j<k\leq n$. Note that $E_{jk}$ refers to the $n\times n$ matrix with only $1$ on the position $(j,k)$ and $0$ elsewhere.
   Consider $\varphi: \bbR^{n^2-1}\rightarrow \fraksu(n)$ the isomorphism denoting linear combinations over the basis. Let $S:=\{\varphi(X_1), \dots, \varphi(X_g)\}$ be a Lie-generating set which has minimal Lie-length. Then there exist brackets $w_1,\dots,w_{n^2-1}$ with length upper bounded by $l$ such that $\{\varphi^{-1}\left[\,w_h\left(\varphi(X_1), \dots, \varphi(X_g)\right)\right]\,|\,h=1,\dots, n^2-1\}$ span the whole $\bbR^{n^2-1}$. We note that this constitutes a polynomial over $g(n^2-1)$ real numbers since the structure constants are either 0 or 1 or -1.  
   
   We define a function $P: \bbR^{g(n^2-1)} \rightarrow \bbR$ 
   \begin{align*}
       P:=\det \left[\,W(w_1,\dots,w_{n^2-1}, \cdot,\cdot,\dots,\cdot,\cdot)\right]\,
   \end{align*}
   where $W$ is the $(n^2-1)$-dimensional matrix with $\varphi^{-1}\left[\,w_h\left(\varphi(X_1), \dots, \varphi(X_g)\right)\right]\,$ as $h$-th column. $P$ is then a polynomial, and since $P(X_1,\dots,X_g)\neq 0$, it is not identically zero. The zero set of $P$ is thus of Lebesgue measure 0. Since the pullback of the measure under $\varphi$ on $\fraksu(n)$ is absolutely continuous with respect to Lebesgue measure on $\bbR^{n^2-1}$, elements $S\in\fraksu(n)^g$ such that $w_h\left(S_1, \dots, S_g\right)$ with $h\in\{1,\dots,n^2-1\}$ are linearly dependent form a measure zero set in $\fraksu(n)^g$. Hence, almost all subsets with $g$ elements have length bounded by $l$. Together with the fact that $l$ is the minimum, the claim follows. 
\end{proof}

We note that the proof strategy can be generalized straightforwardly to various types of matrix Lie algebras. This explains why the numerical simulations always output the minimal length, namely the ``generic'' Lie-length, by taking random generating systems in the numerical tests. Additionally, these numerical experiments show that, in small dimensions, the Lie-length grows as $\Theta(\log n)$ generically. So for a complete analytical proof, the missing ingredient here is finding `good' candidates for $w_1,\dots,w_{n^2-1}$ with an optimal length.

\subsection{Applications}

A better understanding of the Lie-length of a matrix Lie algebra might yield interesting applications in various contexts, such as for the infinitesimal generation of unitary transformations, to simplify the Baker-Campbell Hausdorff formula, or for the ultracontractivity of Hörmander fields. 

\vspace{0.1cm}

\paragraph{Infinitesimal generation of unitary transformations.}
One example of potential applications of a positive solution to our conjecture for the Lie-length of Lie algebras in the generic case is in the context of infinitesimal generation of unitary transformations. Let us consider a unitary transformation $U\in SU(d)$. We generate the unitary by alternatively applying some given Hamiltonians. For example in case of two Hamiltonians $A$ and $B$, such that we assume $-iA$ and $-iB$ generate $\fraksu(d)$, the propagator will be of the form
\begin{align}\label{problemunitarycontrol}
U(\vec{t},\vec{\tau})=e^{-iB\tau_N}e^{-iAt_N}\dots e^{-iB\tau_1}e^{-iAt_1}.
\end{align} 
In \cite{lloyd2019efficient}, the authors observed this model and made the same conjecture on the Lie-length of generic pairs $A$ and $B$ through numeric verification. As described in \cite{lloyd2019efficient}, the Lie-length of these given Hamiltonians as a generating system gives an upper bound on the minimal time required to implement a desired quantum gate while applying sequences of these Hamiltonians, and the overall performance $\ttLie\ell=O(\log d)$ suggests a time scaling as $O(d^{2} \log d)$ for computing the unitary transformation. 

\vspace{0.1cm}

\paragraph{Baker-Campbell-Hausdorff formula.}
Moreover, such results on the Lie-length of a Lie algebra might as well be related to the Baker-Campbell-Hausdorff (BCH) formula. Given a finite-dimensional Lie algebra and two elements $X$ and $Y$, the BCH formula allows us to write $\log (\exp(X)\exp(Y))$ as an infinite series of right-nested commutators of $X$ and $Y$. Therefore, finding a basis of brackets with small length might allow us to find new expressions, which are more easily computable for this series. In theory, it would tell us that we can truncate the infinite series at terms of length up to order $O(\log n)$. However, since our result would not be constructively explicit (as the analogous result for matrix algebras is not either) on how the long right-nested commutators are generated from the shorter ones, this application might not be practically feasible from such an analytical result.

\vspace{0.1cm}

\paragraph{Ultracontractivity of Hörmander fields.}
Last, but not least, we would like to mention the application of Lie-length on ultracontractivity of Hörmander fields. For a compact matrix Lie group/Riemannian manifold $G$ with $n$-dimensional Lie algebra/tangent space $\frakg$, let us consider the sub-Laplacian $\Delta_X= \sum_{i=1}^{k} X_i^*X_i$ with a Lie-generating system $S:=\{X_i|i=1,\dots,k\}$ (in other words, satisfying \textit{Horm\"ander's} condition). Let $P_t=e^{\Delta_X t}$, then
        \begin{align*}
            \|P_t: L_1(G)\rightarrow L_\infty (G)\| \leq C t^{-n\operatorname{Lie}\ell(S)/2} \; \text{ for }0<t<1.
        \end{align*} 
The Lie-length of the system works as a parameter in the subelliptic estimate, while it gives a bound on the distance between the semigroup generated by the sub-Laplacian and some smooth function on $G$. For further details, we refer to Chapter 3 and 4 of \cite{varopoulos_saloff-coste_coulhon_1993} and \cite{gao2020fisher}.

\section{Outlook and Open Problems}

In this paper, we reviewed the problems of estimating the length and the Lie-length of matrix algebras, both in the general and the generic case. We subsequently readapted a result from \cite{KlepSpenko_SweepingWords_2016} to provide a generic version of quantum Wielandt's inequality and applied it in the contexts of injectivity of MPS and PEPS, as well as for bounding the index for eventually full Kraus rank of quantum channels. This result provides a notable improvement in the length of both problems from the generally conjectured order $\Theta (n^2)$ to the actual $\Theta(\log n)$ that happens with probability $1$. Particularly we obtained interesting consequences in the context of injectivity of PEPS, namely that almost all translation invariant PEPS with PBC on a grid with side length of $\Omega(\log n)$ are the unique ground state of a local Hamiltonian, irrespective of the dimension of the grid. This sheds new light on the generic case of one of the most relevant open problems for PEPS as mentioned in \cite{Cirac_2019}. We further discussed the length for matrix Lie algebras similarly and provided numerical evidence for the length of random Lie-generating systems of some typical matrix Lie algebras e.g. $\mathfrak{su}(n)$. 

In this case, however, an analytical proof for the Lie-length appearing with probability one, as well as further research on the general upper bound on the Lie-length, remain missing. Note that a rigorous proof of our conjecture could yield promising applications in various contexts within quantum information theory. 
Additionally, we expect that our work motivates further enriching connections between fundamental questions in the field of algebra and that of quantum information theory.

\vspace{0.3cm}

\noindent {\it Acknowledgments.} 
The authors thank Li Gao and Michael M. Wolf for their comments and suggestions. The authors also thank David P{\'e}rez García and Sirui Lu for the discussion on Wielandt's inequality for PEPS, thank Freek Witteveen, Albert Werner and Blazej Ruba for suggesting the construction of a specific PEPS with small injective region from MPS, and thank Georgios Styliaris for giving feedback on that section. Additionally, they also thank Guillaume Aubrun and Jing Bai for fixing a previous flaw in the proof for the index of primitivity of a generic quantum channel, which is attached in the appendix. This work has been partially supported by the Deutsche Forschungsgemeinschaft (DFG, German Research Foundation) under Germany's Excellence Strategy EXC-2111 390814868 and through the CRC TRR 352. YJ acknowledges support from the TopMath Graduate Center of the TUM Graduate School and the TopMath Program of the Elite Network of Bavaria. 

\bibliographystyle{quantum.bst}

\bibliography{biblio}



\appendix

\section{About the Primitivity Index of a Generic Quantum Channel, by Guillaume Aubrun and Jing Bai}

We recall basic facts about the Zariski topology on a complex vector space $V$. Let $\Pol(V)$ be the set of polynomial functions from $V$ to $\mathbb{C}$.
\begin{enumerate}
\item[(i)] A subset $W \subset V$ is Zariski-closed if there is a subset $S \subset \Pol(V)$ such that
\[ W = \{ x \in V \st P(x) =0 \textnormal{ for every } P \in S \}. \] 
\item[(ii)] Given a subset $U \subset V$, the Zariski-closure of $U$, denoted $\overline{U}^Z$, is the set of $x \in V$ such that any $P \in \Pol(V)$ which vanishes on $U$ vanishes at $x$. We have $\overline{U}^Z=V$ (i.e., $U$ is Zariski-dense) if and only if a polynomial vanishing on $U$ is identically zero.
\item[(iii)] \label{re3} Any linear bijection $f : V \to V$ is a homeomorphism for the Zariski topology, because the map $P \mapsto P\circ f$ is a bijection on $\Pol(V)$. In particular, for every subset $U \subset V$, we have $f(\overline{U}^Z) = \overline{f(U)}^Z$.
\item[(iv)] \label{re4}If a Zariski-closed subset $U \subset V$ has nonempty interior for the usual topology, then $U=V$.
\end{enumerate}

Given $A_1,\dots,A_g$ in $M_n(\bbC)$, the completely positive map~$\mathcal{E}$ defined by 
\begin{equation*}
    \mathcal{E} (X) = \underset{i=1}{\overset{g}{\sum}} A_i X A_i^\dagger
\end{equation*}
is trace-preserving if and only if $S=(A_1,\dots,A_g)$ belongs to the set 
\[ W_{n,g} := \{ (A_1,\dots,A_g) \in M_n(\bbC)^g \, : \, \sum_{i=1}^g A_i^\dagger A_i = \1 \}  . \]
The set $W_{n,g}$ carries a natural probability (see \cite{arveson}). 

Let $P \in \Pol( M_n(\bbC)^g )$ be the polynomial appearing in Step~3 of Theorem 1 and let $Z_{n,g}$ be its zero set. The polynomial $P$ has the property that whenever $S= (A_1,\dots,A_g) \not \in Z_{n,g}$, then 
$\wiel(S) \leq 2 \lceil \log n \rceil$. The goal of this appendix is to prove the following proposition.

\begin{prop} \label{prop:zariski-dense}
The set $W_{n,g}$ is Zariski-dense in $M_n(\bbC)^g$.
\end{prop}

In particular, since $M_n(\bbC)^g \setminus Z_{n,g}$ is a nonempty Zariski-open set, it interesects $W_{n,g}$. Moreover, it follows from \cite[Proposition 2.6]{arveson} that the intersection has measure $1$ with respect to the natural probability measure on $W_{n,g}$. This proves Corollary \ref{cor:generic_Kraus_quantum_channels} from the main text.

The proof of Proposition \ref{prop:zariski-dense} relies on two simple lemmas. We denote by $\T \subset \C$ the set of complex numbers with modulus $1$.

\begin{lem}\label{lem1}
For every integer $n \geq 1$, the set $\T^n$ is Zariski-dense in $\C^n$.
\end{lem}

\begin{proof}
The proof is by induction on~$n$. The $n=1$ case follows since the zero set of a nonzero polynomial is finite, while $\T$ is infinite. Assume the lemma true for some integer $n \geq 1$ and let $P \in \Pol(\C^{n+1})$ be a polynomial vanishing on $\T^{n+1}$. For every $x \in \T$, the polynomial $y \mapsto P(x,y)$ vanishes on $\T^n$ and therefore is the zero polynomial by the induction hypothesis. Therefore, for every $y \in \C^{n}$, 
the polynomial $x \mapsto P(x,y)$ vanishes on $\T$ and therefore is the zero polynomial by the $n=1$ case. It follows that $P$ is the zero polynomial, and therefore~$\T^{n+1}$ is Zariski-dense in $\C^{n+1}$.
\end{proof}

\begin{lem} \label{lem2}
For every integer $n \geq 1$, the unitary group $U_n$ is Zariski-dense in $M_n(\C)$.    
\end{lem}

\begin{proof}
Suppose that $P \in \Pol(M_n(\C))$ vanishes on $U_n$. In particular, $P$ vanishes on the set of diagonal matrices with diagonal elements in $\T$. By Lemma \ref{lem1}, $P$ vanishes on the set of all diagonal matrices.

Let $T \in M_n(\C)$. By the singular value decomposition, we may find $U$, $V \in U_n$ and a diagonal matrix $\Sigma$  such that $T = U \Sigma V$. 
Since the map $\rho$ on $M_n(\C)$ defined by 
\[\rho: M \mapsto UMV \]
is a linear bijection, it follows from (iii) that
$\rho(\overline{U_n}^Z)=\overline{\rho(U_n)}^{Z}$
and thus $T \in \overline{U_n}^Z$. Hence the Zariski-closure of $U_n$ is $M_n(\C)$. 
\end{proof}

\begin{proof}[Proof of Proposition \ref{prop:zariski-dense}]
The proof is by induction on $g$. The case $g=1$ follows from Lemma \ref{lem2} since $W_{n,1}$ identifies with~$U_n$.
Now assume that $W_{n,g-1}$ is Zariski-dense in $M_n(\C)^{g-1}$ for some integer $g \geq 2$. Let $M\in M_n(\C)$ such that $0<\Vert M \Vert< 1$ and 
\begin{multline*}
W_{M} := \{ (A_1,\dots,A_{g-1}) \in M_n(\C)^{g-1} \st \\ A_1^*A_1 + \dots + A_{g-1}^*A_{g-1}+ M^*M =\mathds{1} \} .
\end{multline*}

The matrix $C:=\mathds{1}-M^*M$ is positive definite. The map defined on $M_n(\C)^{g-1}$ by 
\[\Phi : (X_1,\dots, X_{g-1}) \mapsto (X_1C^{\frac{1}{2}}, \dots, X_{g-1}C^{\frac{1}{2}} ) \] 
is a linear bijection and thus a homeomorphism for the Zariski topology. Since $\Phi(W_{n,g-1})=W_{M}$, it follows from the induction hypothesis and (iii) that $W_{M}$ is Zariski-dense in~$M_n(\C)^{g-1}$.

On the other hand,  the set $\{ (A_1,\dots,A_{g-1}) \in M_n(\C)^{g-1} \st (A_1,\dots,A_{g-1}, M) \in \overline{W_{n,g}}^{Z}\}$ is Zariski-closed in $M_n(\C)^{g-1}$ and contains $W_{M}$, which implies that it is $M_n(\C)^{g-1}$. The subset 
\[\{ (A_1,\dots,A_g) \in M_n(\C)^{g} \st 0<\Vert A_g \Vert< 1\}\]
is open for the usual topology, and contained in $\overline{W_{n,g}}^Z$. It follows from (iv) that $W_{n,g}$ is Zariski-dense in $M_n(\C)^{g}$. This completes the proof.
\end{proof}

\end{document}